\documentclass[a4paper,cleveref,english,nolineno]{gd-lipics-v2}
\usepackage[T1]{fontenc}
\usepackage[utf8]{inputenc}

\usepackage{amsfonts,amsthm,amsmath,amssymb}
\usepackage{cite,microtype,graphicx,hyperref}

\title{Stabbing Faces By a Convex Curve}
\author{David Eppstein}{Computer Science Department, University of California, Irvine}{eppstein@uci.edu}{}{Research supported in part by NSF grant CCF-2212129.}
\authorrunning{D. Eppstein}
\Copyright{David Eppstein}



\category{Short (Theory)}

\begin{document}
\maketitle

\begin{abstract}
We prove that, for every plane graph $G$ and every smooth convex curve $C$ not on a single line, there exists a straight-line drawing of $G$ for which every face is crossed by $C$.
\end{abstract}

\section{Introduction}
Every smooth convex curve $C$, not on a single line, contains arbitrarily many points in general position (its tangent points at distinct slopes within a suitable interval). It follows that every outerplanar graph has a straight-line drawing with all vertices on $C$~\cite{GriMohPac-AMM-91}. A graph that can be drawn in this way must be outerplanar, as its vertices all belong to a common face outside $C$. But instead of asking for $C$ to contain all vertices, here we ask: which planar graphs can be drawn so that $C$ intersects all faces? The answer turns out to be all of them.

Our motivation for this question comes from the problem of characterizing universal point sets, sets of points $U_n$ for each $n$ having the property that every $n$-vertex planar graph has a straight-line drawing with its vertices in $U_n$. These sets have size $O(n^2)$~\cite{deFPacPol-Comb-90,Sch-SODA-90,Bra-ICTGGT-08,BanEppChe-JGAA-14}, and $\ge cn$ for a constant $c>1$~\cite{ChrPay-IPL-95,Kur-IPL-04,SchSchSte-JGAA-20}. A recent line of research has investigated the combinatorial structure of universal point sets, in terms of the number of lines necessary to cover all points of such a set. Although point sets on $O(1)$ parallel lines can support only the planar graphs of bounded pathwidth~\cite{DujFelKit-Algo-08}, two crossed lines can support a much wider class of graphs, including all leveled planar graphs~\cite{BanDevDuj-Algo-19}. However, supporting all $n$-vertex planar graphs requires an unbounded number of lines, $\Omega(n^{1/3})$~\cite{TavVer-WG-11,ChaFleLip-JoCG-20,Epp-JoCG-21}. Eppstein~\cite{Epp-18} asked: what if we replace lines in these results by convex curves? Is it possible that universal point sets can be supported by only $O(1)$ convex curves?

This question could be answered by finding a single planar graph $G^?$ that could not be drawn with all of its faces crossed by a single convex curve. If such a graph $G^?$ existed, we could assume without loss of generality that $G^?$ is a triangulation, and then recursively replace triangular faces of $G^?$ by more copies of $G^?$. This recursive replacement process would generate a family of graphs in which, for every drawing, every convex curve touches a number of faces bounded by a sublinear polynomial of the number of vertices. This would imply that universal point sets for graphs in this family could be supported only by a polynomially growing number of convex curves. However, the hope for a proof along these lines is dashed by our result: $G^?$ does not exist.

We remark that the corresponding question for edges rather than faces has an intermediate answer: the class of graphs that can be drawn with all edges touching a smooth convex curve $C$ is broader than the outerplanar graphs but less broad than the planar graphs. In \cref{sec:edges}, we provide an example of a planar graph that cannot be drawn in this way, regardless of the choice of $C$. It takes the form of the graph of nine regular octahedra, with the outer eight octahedra each glued to a separate face of a central octahedron.  But unlike the case for faces crossed by $C$, we do not know how to leverage this example to prove anything about the number of convex curves needed to support a universal point set.

\section{Preliminaries}

Given a plane graph $G$,\footnote{By a \emph{plane graph} we mean a planar graph together with a combinatorial embedding, which should be respected in our drawing of the graph.} to be drawn with all faces crossed by a convex curve, we may assume without loss of generality that $G$ is maximal planar, because if we complete any other graph to a maximal planar graph, and draw the result with $C$ crossing all faces of the maximal planar completion, then $C$ will also cross all faces of the original graph.

For a maximal plane graph $G$, with a fixed choice of outer face, we will draw $G$ using a \emph{canonical ordering}~\cite{Kan-Algo-96}, as used by de Fraysseix, Pach, and Pollack to find grid drawings of planar graphs~\cite{deFPacPol-Comb-90}. This is an ordering of the vertices of $G$ such that the first three vertices in the ordering induce a triangle with one edge (the \emph{base edge}) on the outer face, and such that each prefix of the ordering induces a triangulated disk\footnote{By a \emph{triangulated disk}, we mean a plane graph in which the outer face is a simple polygon and all interior faces are triangles.} in the drawing of $G$, with the bounded faces of this triangulated disk also being faces of the full drawing of $G$. The next vertex in the ordering, after any prefix, must have as its earlier neighborhood a contiguous path along the boundary of the disk, not including the base edge, and the region between this path and the next vertex is triangulated in a fan of triangles meeting at the added vertex. This structure lends itself well to greedy incremental algorithms for planar graph drawing and to induction proofs of properties of these drawings.

For every positive integer $k$, there exist maximal planar graphs in which, for every drawing, the number of faces of the graph that can be crossed by a line is less than a $1/k$ fraction of all faces~\cite{GruWal-JCTA-73}. These graphs therefore cannot have all faces crossed by a $k$-gon. To avoid these examples, our results involve smooth convex curves. There are many ways of defining curves and their smoothness; the properties we require are that the curve must form a connected subset of the boundary of a convex set in the plane (for instance, its convex hull), it must have a unique tangent line at each point, and the slope of the tangent line must be continuous along the curve ($C^1$~smoothness, meaning that the first derivative is continuous). We do not require strict convexity: our curves may contain line segments (along which the tangent slope is constant). For instance, a \emph{stadium}, a convex shape obtained by attaching semicircular end caps to two opposite sides of a rectangle, has a smooth convex boundary in this sense. However, a polygon is not smooth at its vertices. An \emph{arc} of a convex curve is any connected proper subset. The \emph{total curvature} of an arc can be measured by the change in slopes of the tangent lines along the arc, as an angle.

With a stronger smoothness assumption of being twice continuously differentiable ($C^2$~smoothness), the ability to cross all faces of a planar drawing would extend to non-convex smooth curves that do not lie on a line, because each such a curve has a smooth convex arc (within a neighborhood of any point of nonzero curvature) to which we can apply our proof. We omit the details.

\section{Adequate drawings}

We will separate out into lemmas the steps of an induction proof that a given plane graph $G$ has a drawing with all faces crossed by a given smooth convex curve $C$. If we reinterpret the proof as a greedy construction algorithm, each step of the construction has the following property.

\begin{figure}[t]
\centering\includegraphics[scale=0.25]{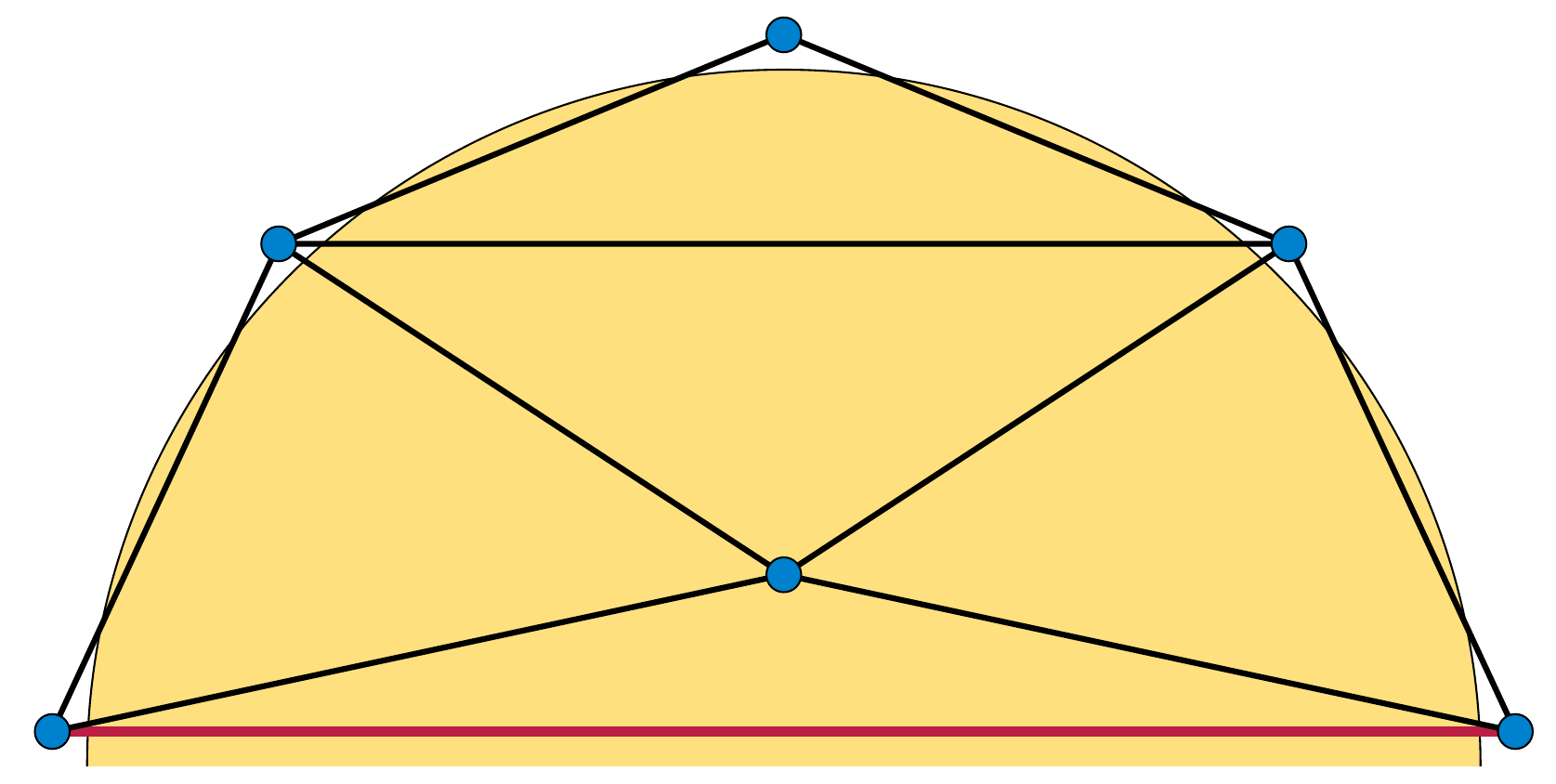}
\caption{Adequate drawing of a triangulated disk (base edge shown in red) with respect to a semicircle: all boundary vertices are exterior to the semicircle and all non-base boundary edges are crossed, twice each, by disjoint arcs of the semicircle.}
\label{fig:adequate}
\end{figure}

\begin{definition}
Given an abstract triangulated disk $D$ with a specified base edge $B$, and a smooth convex curve $C$ that is not on a single line, we define an \emph{adequate drawing} of $D$ to be a straight-line drawing with the following properties:
\begin{itemize}
\item Every boundary vertex of $D$ is exterior to the convex hull of $C$.
\item Every boundary edge $e$ of $D$ other than $B$ is crossed twice by $C$. The two crossing points bound an arc of $C$ that is disjoint from $D$ and from the other arcs defined in the same way from other boundary edges of $D$.
\item The crossings of $C$ with $D$ span an arc of $C$ that bends through a total curvature of less than $\pi$.
\end{itemize}
\end{definition}

\cref{fig:adequate} depicts an example. Not all faces in this example are crossed by the convex curve $C$ (a semicircle), but the faces incident to the non-base boundary edges are crossed. As the name suggests, it will be adequate to prove that every triangulated disk formed by the canonical ordering has an adequate drawing, as every face will be adjacent to the boundary for the induced disk at some point in the ordering. This ensures the property we are really trying to prove, that the whole drawing of $G$ has every face crossed by $C$:

\begin{lemma}
\label{lem:stab}
Let $G$ be a maximal plane graph $G$ with $n$ vertices, with a straight-line drawing in the plane, and let $C$ be a smooth curve that is not on a single line.
Suppose that $G$ has a canonical ordering with the property that, for each triangulated disk $D_i$ (with $3\le i\le n$) induced by the first $i$ vertices in the canonical ordering, the drawing of $G$ restricted to $D_i$ is adequate. Then the drawing of $G$ has all faces crossed by $C$.
\end{lemma}

\begin{proof}
Each face $\Delta$ of $G$ has at least one edge $e_\Delta$ that is a boundary edge of some $D_i$, but not the base edge of the canonical ordering.
For the outer face, $e_\Delta$ can be any non-base edge of $D_n$; for any other face, $e_\Delta$ is the edge of $\Delta$ induced by the first two vertices of $\Delta$ to appear in the canonical ordering. Because it is a non-base boundary edge of $D_i$, and because of the assumption that the drawing of $D_i$ is adequate, $e_\Delta$ is crossed (twice) by $C$. Near these edge crossing points, $\Delta$ itself is crossed by $C$.
\end{proof}

\section{Induction proof}

We will prove our main result by induction, using canonical orderings. The next lemma is the base case for this induction proof.

\begin{figure}[t]
\centering\includegraphics[scale=0.25]{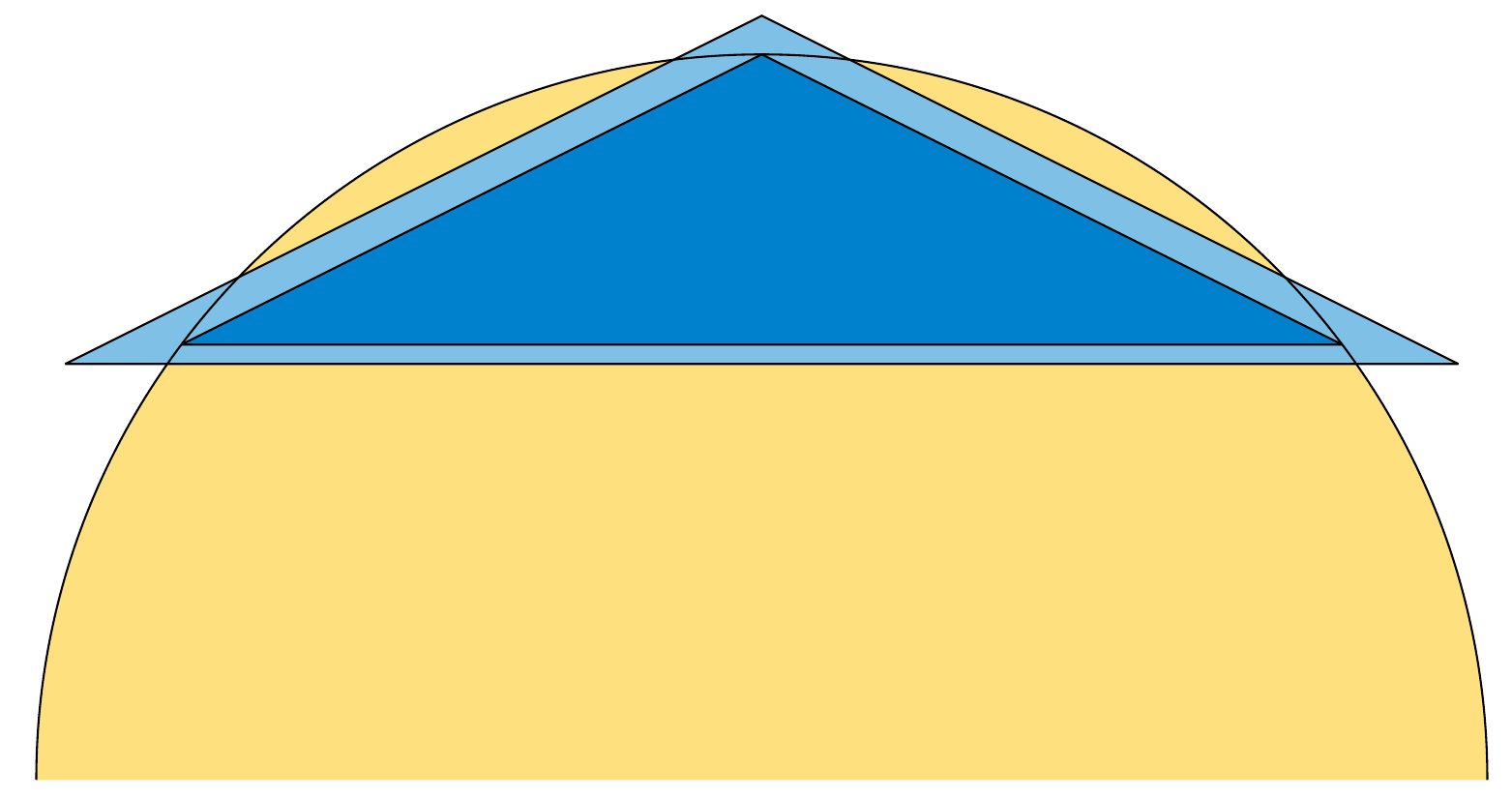}
\caption{Base case (\cref{lem:base}): constructing an adequate drawing of a triangle by scaling an inscribed triangle with three distinct tangent slopes.}
\label{fig:base}
\end{figure}

\begin{lemma}
\label{lem:base}
For every smooth convex curve $C$ that does not lie on a line, there exists an adequate drawing of a triangle with a designated base edge.
\end{lemma}

\begin{proof}
Let $A$ be any arc of $C$ of nonzero curvature less than $\pi$, and inscribe a triangle in $A$ so that the tangent lines to $A$ at its three vertices have distinct slopes; this is possible because, by smoothness and nonzero curvature, $A$ has tangents with a continuous range of slopes.  Choose the base edge of this triangle to be the edge connecting the two extreme vertices in their ordering along $A$. Scale the triangle around its centroid by a factor $1+\varepsilon$ for a parameter $\varepsilon>0$ chosen sufficiently small,

By convexity, the inscribed triangle lies in the closed convex hull of $C$, and no point in the interior of one of its edges can lie on curve $C$ itself: if such a point existed then the convexity of $C$ would force $C$ to contain that edge, violating the assumption of distinct tangent slopes at the three vertices. Scaling brings the triangle vertices exterior to $C$, but for sufficiently small choices of $\varepsilon$, some points interior to each non-base edge remain interior to the convex hull of $C$ (\cref{fig:base}). Then each non-base edge of the scaled triangle will have two crossings with $C$ separating its endpoints (exterior to the hull) with the points along the edge interior to the hull. These crossings will separate disjoint arcs of $C$, and together span an arc of $C$ that is a sub-arc of $A$, so all the requirements of an adequate drawing are met.
\end{proof}

The simpler of the two inductive cases involves adding a single triangle as an ear of the triangulated disk, attached to it by a single boundary edge of the previous disk.

\begin{figure}[t]
\centering\includegraphics[scale=0.25]{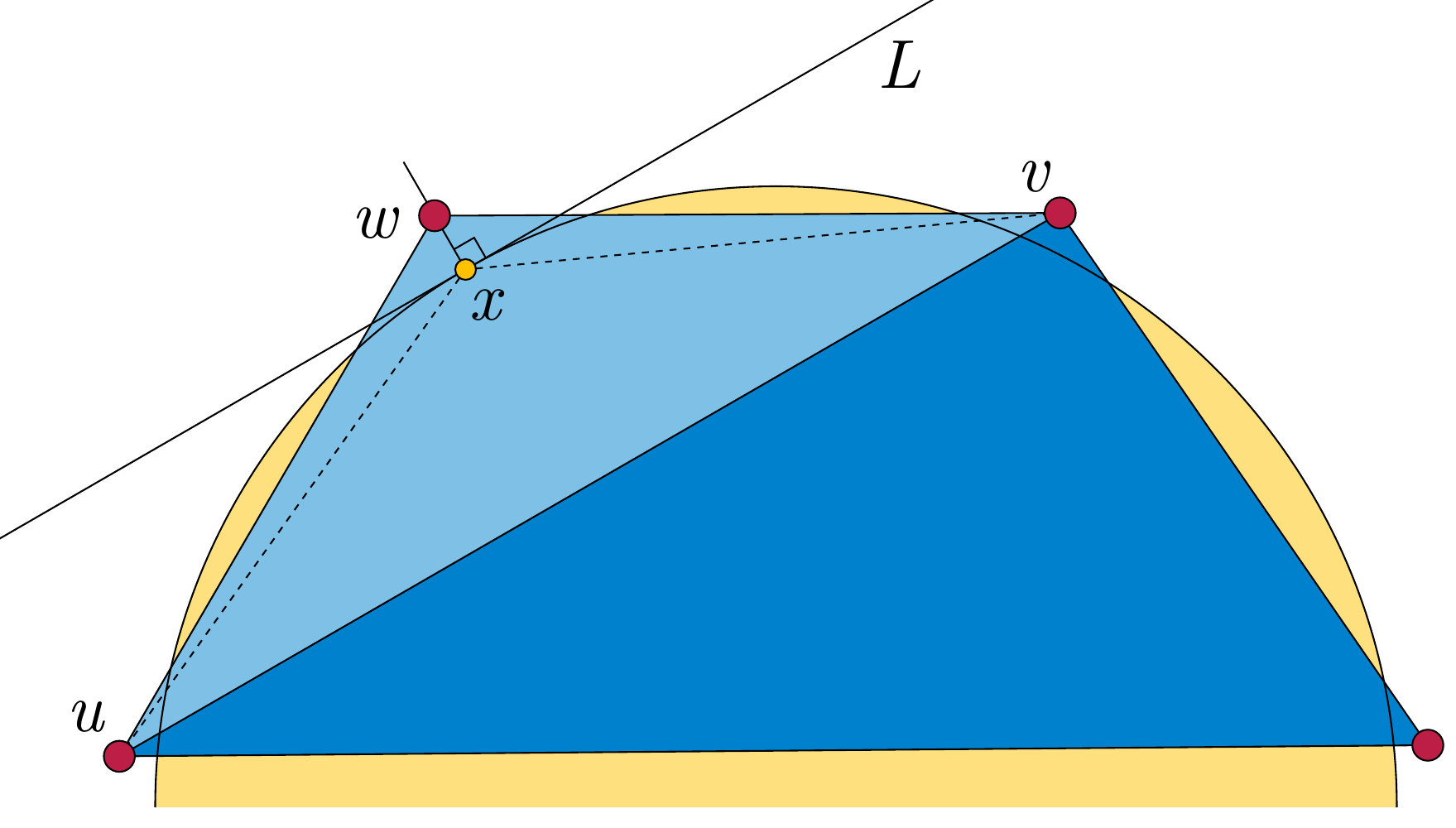}
\caption{Ear case (\cref{lem:ear}): attaching triangle $uvw$ to edge $uv$ of an adequate drawing.}
\label{fig:ear}
\end{figure}

\begin{lemma}
\label{lem:ear}
Let $D$ be a triangulated disk with an adequate drawing with respect to a smooth convex curve $C$, and let $uv$ be any boundary edge of the drawing of $D$ that is not the base edge of the drawing. Then it is possible to add a vertex $w$ adjacent only to $u$ and $v$, forming a triangle attached to $D$ across edge $uv$, and to place $w$ in the drawing so that it becomes an an adequate drawing of the resulting augmented triangulated disk.
\end{lemma}

\begin{proof}
For the notation used here, see \cref{fig:ear}. By the assumption that the drawing of $D$ is adequate, an arc of curve $C$ crosses twice through edge $uv$. By smoothness, this arc has a tangent line $L$ parallel to $uv$, touching the arc at a point $x$ between the two crossings of $C$ with $uv$. Polyline $uxv$ starts outside $C$ at $u$, crosses inside $C$ to touch $C$ at $x$ from the interior of its convex hull, and then crosses $C$ again to reach the point $v$ exterior to $C$. Place $w$ on a perpendicular line to $L$ through $x$, outward from $C$ at a sufficiently small distance $\varepsilon$ from $x$.

If $\varepsilon$ is sufficiently small, the polyline $uwv$ crosses $C$ from exterior to interior at a point near the crossing of $C$ with line segment $ux$, and crosses again from interior to exterior at a point near the crossing of $C$ with line segment $xv$. In order for polyline $uwv$ to be exterior to $C$ at $w$, between these two crossings, each of the two new edges $uw$ and $wv$ must have a second crossing near $w$. Thus, the requirements of an adequate drawing at the new boundary edges $uw$ and $wv$ are met, and the crossings of $C$ with the other boundary edges of $D$ remain unaffected, so the result is an adequate drawing.
\end{proof}

In the remaining inductive case, we add a fan of triangles, connecting an added vertex to a path of two or more edges of the triangulated disk.

\begin{lemma}
\label{lem:fan}
Let $D$ be a triangulated disk with an adequate drawing with respect to a smooth convex curve $C$, and let $P=v_0v_1\dots v_k$ be a path of $k$ boundary edges of the drawing of $D$ (for $k\ge 2$) that does not pass through the base edge of the drawing. Then it is possible to add a vertex $w$ adjacent to each vertex $v_i$ of the path, forming a fan of triangles attached to $P$ across edge $uv$, and to place $w$ in the drawing so that it becomes an adequate drawing of the resulting augmented triangulated disk.
\end{lemma}

\begin{proof}
For the notation used here, see \cref{fig:fan}. Draw a ray from vertex $v_0$ through vertex $v_1$, and a second ray from vertex $v_k$ through vertex $v_{k-1}$. Each ray crosses $C$ at the two crossing points of an edge of $P$. By the assumption that the drawing is adequate and therefore its crossings with $C$ span an arc of $C$ that bends through an angle less than $\pi$, the two rays cross rather than diverging. Let $x$ be the crossing point of the two rays. At $x$, the two crossing rays subdivide the plane into four wedges. Place $w$ interior to the wedge opposite from $P$, at a sufficiently small distance from $w$.

Vertices $v_1,\dots v_{k-1}$ all lie in the wedge containing $P$, and are visible to $x$ with respect to the given drawing of $D$; however, the two endpoints $v_0$ and $v_k$ of the path are blocked from visibility by $v_1$ and $v_{k-1}$ respectively. Placing $w$ anywhere in the wedge opposite from $P$ preserves the visibilities with $v_1,\dots v_{k-1}$ and allows $w$ to also see the two path endpoints, producing a valid drawing of $D$. The segments $v_0x$ and $xv_k$ cross $C$ twice, along edges $v_0v_1$ and $v_{k-1}v_k$ (which lie on these segments), and if $w$ is sufficiently close to $x$ then the segments $v_0w$ and $wv_k$ will also cross $C$ twice, meeting the requirements of an adequate drawing at the two new boundary edges of the augmented drawing.  The crossings of $C$ with the other boundary edges of $D$ remain unaffected, so the result is an adequate drawing.
\end{proof}

\begin{figure}[t]
\centering\includegraphics[scale=0.25]{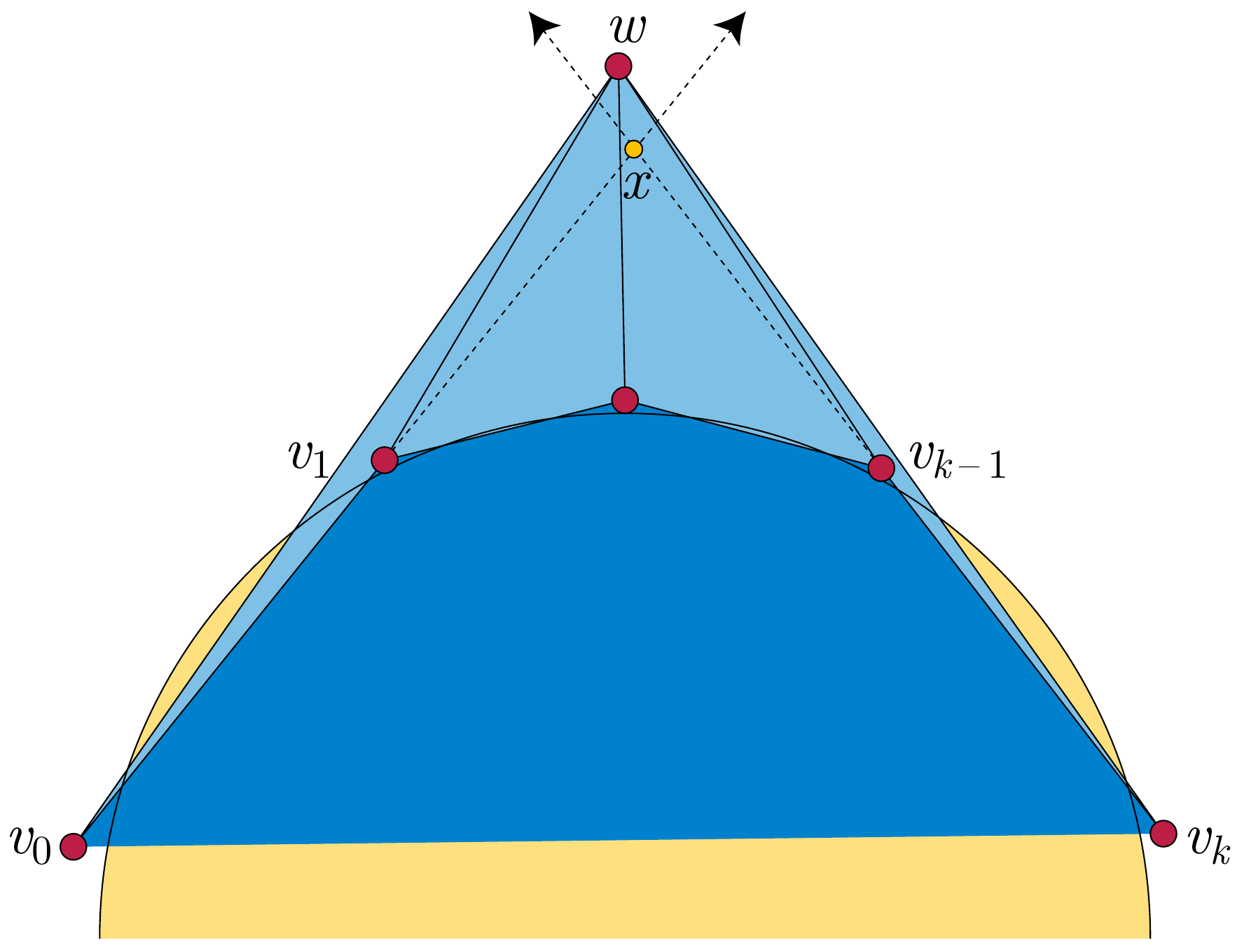}
\caption{Fan case (\cref{lem:fan}): attaching a fan of triangles to path $P=v_0v_1\dots v_k$ of an adequate drawing.}
\label{fig:fan}
\end{figure}

With the lemmas above, we are ready to prove our main result:

\begin{theorem}
Let $G$ be a plane graph, and let $C$ be a smooth convex curve that does not lie on a line. Then $G$ has a straight-line drawing in which all faces are crossed by $C$.
\end{theorem}

\begin{proof}
Augment $G$ to a maximal plane graph, and construct a canonical ordering. We prove by induction that there exists a drawing $D$ of the augmented graph such that, for each triangulated disk $D_i$ (with $3\le i\le n$) induced by the first $i$ vertices in the canonical ordering, the drawing restricted to $D_i$ is adequate. The base case $i=3$ is \cref{lem:base}, and the induction step for each $i>4$ is either \cref{lem:ear} (if vertex $i$ has two neighbors in $D_i$) or \cref{lem:fan} (if vertex $i$ has three or more neighbors in $D_i$. This proves that $D$ meets the hypotheses of \cref{lem:stab}, so every face of the augmented graph is crossed by $C$. Each face of $G$ itself is a union of faces of the augmented graph, so it is also crossed by $C$.
\end{proof}

\begin{figure}[t]
\centering\includegraphics[scale=0.25]{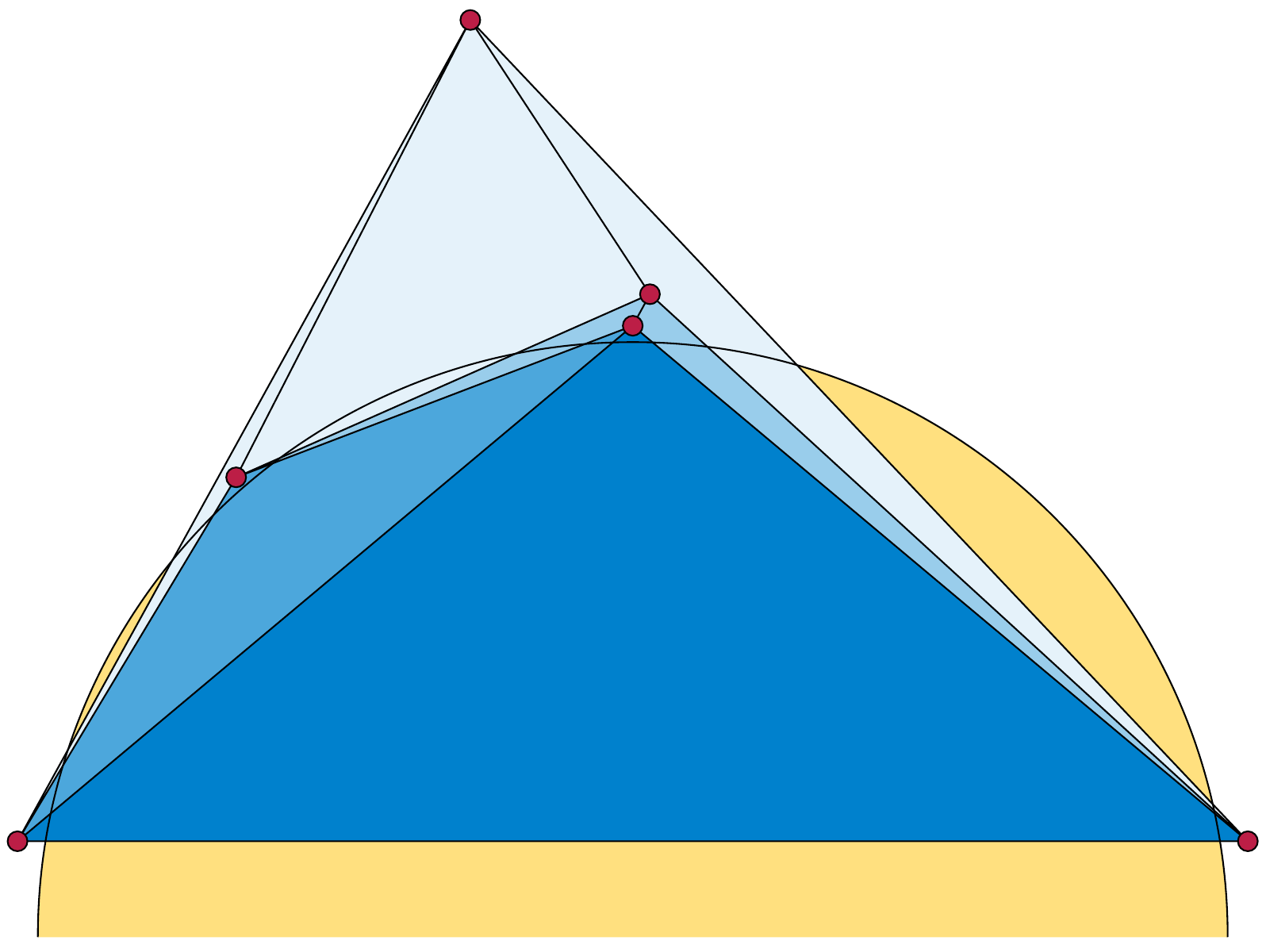}
\caption{Drawing of the graph of an octahedron with all faces crossed by a semicircle. The shading indicates the order in which triangles were added to the drawing in its canonical ordering, with darker triangles added earlier than lighter ones.}
\label{fig:oct-stabbed}
\end{figure}

Our induction proof can be turned into an algorithm in which we add vertices one at a time to a drawing, using the construction of \cref{lem:base} to add the first three vertices, and either \cref{lem:ear} or \cref{lem:fan} to add each subsequent ear. An example of the resulting drawing, for the graph of an octahedron with its faces crossed by a semicircle, is shown in \cref{fig:oct-stabbed}. However, to analyze the numerical precision of vertex coordinates needed for this construction, and the angular resolution of the result, we would need additional assumptions about $C$. As the figure shows, even for a well-behaved curve such as a semicircle, the angular resolution can be quite low: this construction is more useful in producing counterexamples than as a method for producing readable graph drawings.

\section{Conclusions and open problems}
We have shown that for every planar graph and every smooth convex curve, there exists a drawing of the graph in which all faces are crossed by the curve.
The most salient remaining open problem is the one that motivated this research: is there a constant $k$ such that all planar graphs can be drawn with their vertices on $\le k$ convex curves? It would also be of interest to more precisely characterize the graphs that can be drawn with all edges crossed by a convex curve (for which see \cref{sec:edges}) and to extend these results to beyond-planar graphs.

\bibliographystyle{plainurl}
\bibliography{stabbing}

\newpage
\appendix

\section{A graph whose edges cannot be touched by a convex curve}
\label{sec:edges}

We have seen that the graphs that have planar drawings where all vertices are touched by a given convex curve are exactly the outerplanar graphs,
and that all planar graphs have planar drawings where all faces are crossed by a given convex curve. This naturally raises the question of which planar graphs have planar drawings where all edges are touched or crossed by a given convex curve. We do not settle this question, but in this section we provide an example showing that it is nontrivial: not every planar graph has such a drawing. In fact, there exist planar graphs that do not have any such drawing, regardless of which convex curve we choose.

First, for the version of the question where we ask for all edges to be crossed by a convex curve, consider the graph $K_{2,2,2}$, the graph of the regular octahedron. This graph can have all edges touched by a convex curve, with some touches at the endpoint of an edge (\cref{fig:oct-edges}) but, as we prove, it is not possible for a convex curve to cross all edges at points in the relative interior of the edge.

\begin{figure}[h!]
\centering\includegraphics[scale=0.25]{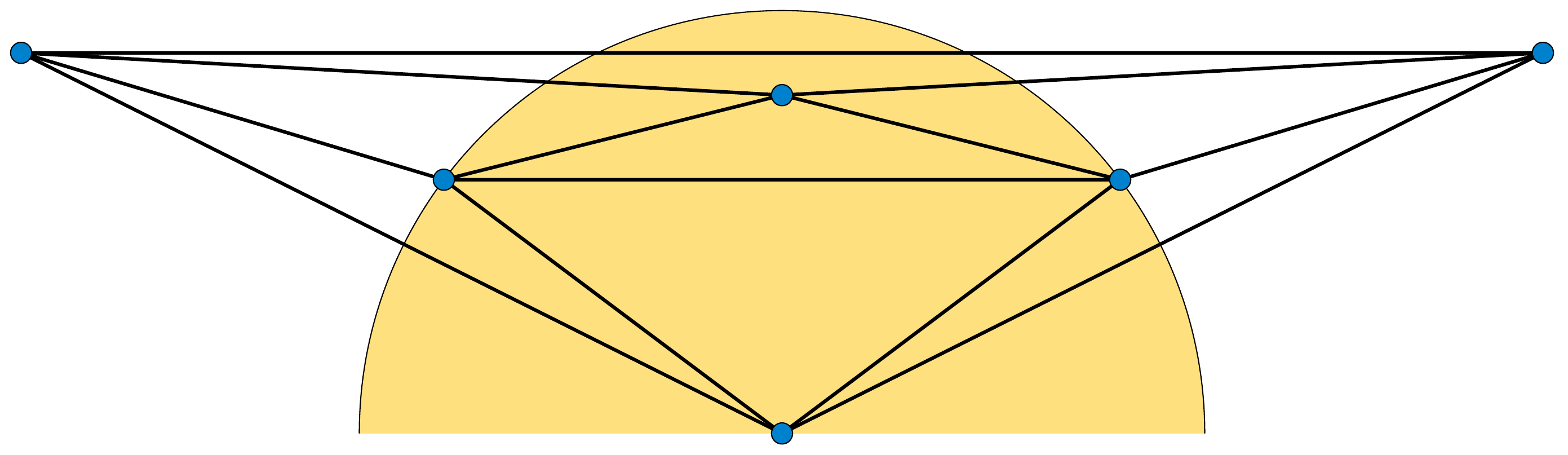}
\caption{The octahedral graph $K_{2,2,2}$ drawn with all edges touched by a semicircle.}
\label{fig:oct-edges}
\end{figure}

\begin{observation}
\label{obs:oct}
$K_{2,2,2}$ does not have a planar straight-line drawing for which there exists a convex curve $C$ that crosses all edges.
\end{observation}

\begin{proof}
Any drawing of this graph must consist of two nested triangles, connected to each other by six edges.
We consider the following cases for how $C$ can be positioned with respect to these triangles:
\begin{itemize}
\item If the inner triangle vertices are all on or interior to $C$, the inner triangle is uncrossed. In this case, $C$ can touch all edges (at the endpoints of the three edges of the inner triangle) but it cannot cross the inner three edges.
\item Otherwise, at least one vertex $v$ of the inner triangle is exterior to $C$. By the convexity of $C$ and by Minkowski's hyperplane separation theorem, there exists a line $\ell$ through $v$ disjoint from $C$. The three triangle faces of $K_{2,2,2}$ each span an angle less than $\pi$ at $v$, so the halfplane bounded by $\ell$ that does not contain $C$ must include portions of at least two of these faces. The edge separating these two faces is separated by $\ell$ from $C$, so $C$ does not touch this edge.\qedhere
\end{itemize}
\end{proof}

Now form a 30-vertex planar graph $G_{30}$ by replacing each face of an octahedron by a copy of the octahedral graph. This can be represented geometrically either as a convex polyhedron~(\cref{fig:octoct}) or as a compound of nine octahedra formed by gluing one face of an octahedron to each face of a central octahedron.

\begin{figure}[t]
\centering\includegraphics[width=0.5\textwidth]{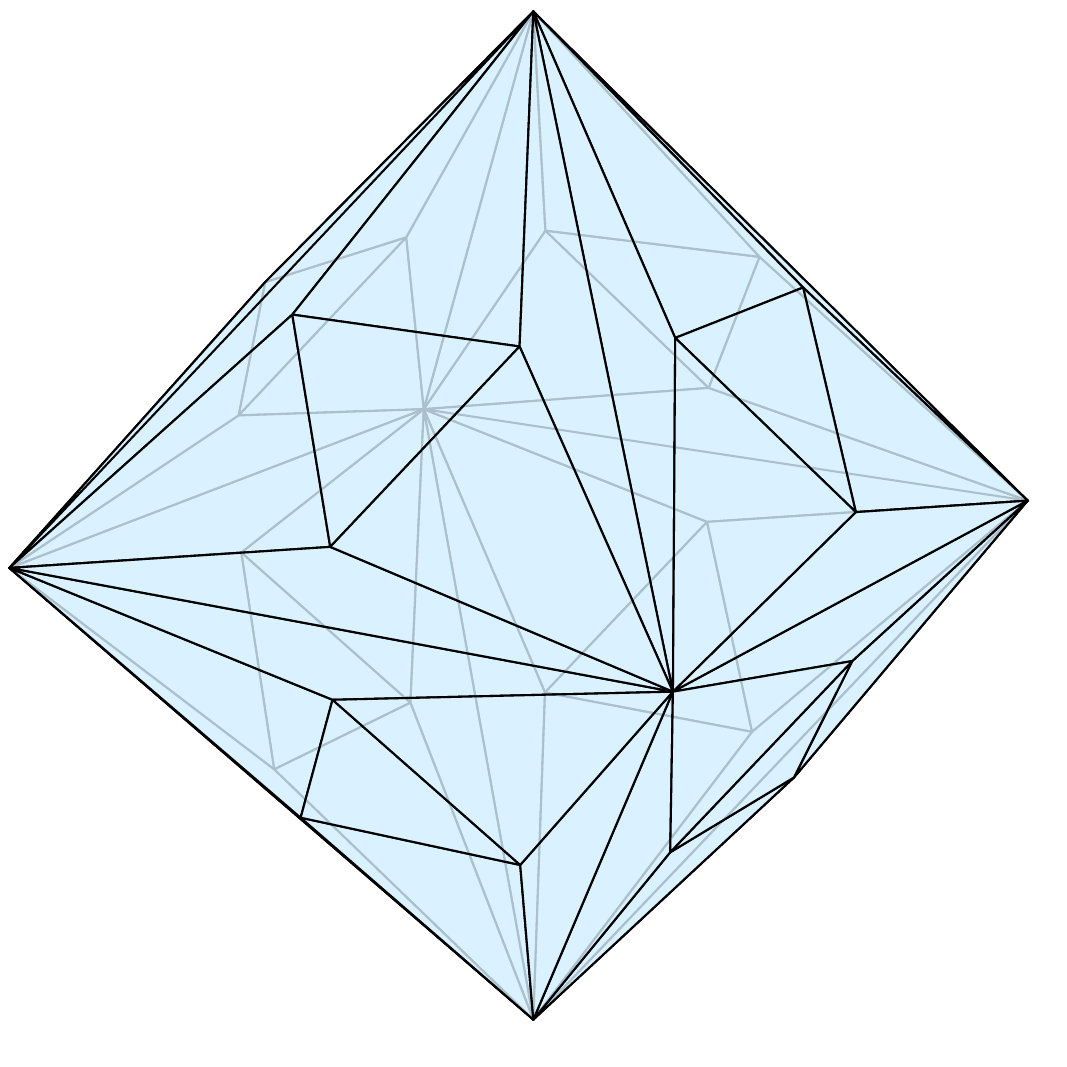}
\caption{The graph $G_{30}$, an octahedron with its faces replaced by nested copies of the octahedral graph.}
\label{fig:octoct}
\end{figure}

\begin{observation}
$G_{30}$ does not have a planar straight-line drawing for which there exists a convex curve $C$ that touches all edges.
\end{observation}

\begin{proof}
Such a drawing would contain within it a drawing of the central octahedral graph $K_{2,2,2}$ (to which eight others have been glued).
By the case analysis of \cref{obs:oct}, this central octahedron must be drawn as two nested triangles, and in the only case in which $C$ can touch all of its edges, the inner of these two nested triangles must be drawn with all of its vertices on or interior to $C$. In $G_{30}$, this inner triangle has been replaced by another copy of $K_{2,2,2}$, whose drawing is entirely on or interior to $C$. The inner triangle of this copy is entirely interior to $C$, and its edges cannot be touched by $C$.
\end{proof}

\end{document}